\documentclass[9pt,shortpaper,twoside,web]{ieeecolor}
\usepackage{generic}
\usepackage{cite}
\usepackage{amsmath,amssymb,amsfonts}
\usepackage{algorithmic}
\usepackage{graphicx}
\usepackage{textcomp}

\usepackage{amsbsy}
\usepackage{fixmath}
\usepackage{bm}

\newtheorem{theorem}{Theorem}
\newtheorem{lemma}{Lemma}

\def\BibTeX{{\rm B\kern-.05em{\sc i\kern-.025em b}\kern-.08em
    T\kern-.1667em\lower.7ex\hbox{E}\kern-.125emX}}
\begin{document}
\title{Application of Terminal Region Enlargement Approach for Discrete Time Quasi Infinite Horizon NMPC} 
\author{Chinmay Rajhans, \IEEEmembership{Member, IEEE}, and Sowmya Gupta, \IEEEmembership{Member, IEEE}
\thanks{``This work is not supported by any funding agency.'' }
\thanks{C. Rajhans is with Department of Electrical Engineering, IIT Bombay, Mumbai, India (e-mail: rajhanschinmay2@gmail.com). }
\thanks{S. Gupta is with Center for Research in Nano Technology and Science, IIT Bombay, Mumbai, India (e-mail: sowmya.gupta283@gmail.com). }
}
%\thanks{This paragraph of the first footnote will contain the date on 
%which you submitted your paper for review. It will also contain support 
%information, including sponsor and financial support acknowledgment. For 
%example, ``This work was supported in part by the U.S. Department of 
%Commerce under Grant BS123456.'' }
%\thanks{The next few paragraphs should contain 
%the authors' current affiliations, including current address and e-mail. For 
%example, F. A. Author is with the National Institute of Standards and 
%Technology, Boulder, CO 80305 USA (e-mail: author@boulder.nist.gov). }
%\thanks{S. B. Author, Jr., was with Rice University, Houston, TX 77005 USA. He is 
%now with the Department of Physics, Colorado State University, Fort Collins, 
%CO 80523 USA (e-mail: author@lamar.colostate.edu).}
%\thanks{T. C. Author is with 
%the Electrical Engineering Department, University of Colorado, Boulder, CO 
%80309 USA, on leave from the National Research Institute for Metals, 
%Tsukuba, Japan (e-mail: author@nrim.go.jp).}}

\maketitle

\begin{abstract} 
Ensuring nominal asymptotic stability of the Nonlinear Model Predictive Control (NMPC) controller is not trivial. Stabilizing ingredients such as terminal penalty term and terminal region are crucial in establishing the asymptotic stability. Approaches available in the literature provide limited degrees of freedom for the characterization of the terminal region for the discrete time Quasi Infinite Horizon NMPC (QIH-NMPC) formulation. Current work presents alternate approaches namely arbitrary controller based approach and LQR based approach, which provide large degrees of freedom for enlarging the terminal region. Both the approaches are scalable to system of any state and input dimension. Approach from the literature provides a scalar whereas proposed approaches provide a linear controller and two additive matrices as tuning parameters for shaping of the terminal region. Proposed approaches involve solving modified Lyapunov equations to compute terminal penalty term, followed by explicit characterization of the terminal region. Efficacy of the proposed approaches is demonstrated using benchmark two state system. Terminal region obtained using the arbitrary controller based approach and LQR based approach are approximately 10.4723 and 9.5055 times larger by area measure when compared to the largest terminal region obtained using the approach from the literature. 
\end{abstract}

\begin{IEEEkeywords}
Asymptotic stability, Control systems, Lyapunov methods, Optimal control, Predictive Control.  
%Enter key words or phrases in alphabetical 
%order, separated by commas. For a list of suggested keywords, send a blank 
%e-mail to keywords@ieee.org or visit \underline
%{http://www.ieee.org/organizations/pubs/ani\_prod/\discretionary{}{}{}keywrd98.txt}
\end{IEEEkeywords}

\section{Introduction}
\label{sec:introduction}

\IEEEPARstart{M}{odel} Predictive Control (MPC) finds its applications in every field of engineering \cite{Qin2003, Camacho2007}. Principle concept for ensuring nominal and robust stability involves inclusion of stabilizing constraints \cite{Rawlings1993}. A significant process has taken place in the area of nominal and robust stability of linear MPC and Nonlinear MPC (NMPC) \cite{Mayne2000, Mayne2014, Rawlings2017}.  Terminal equality constraint is mathematically equivalent to equating the terminal state to the target set point or target steady state operating point \cite{Keerthi1988}. Primary limitation of terminal equality constraint is that it is very conservative and often leads to infeasibility specifically for constrained formulations. Michalska and Mayne conceptualized dual mode MPC scheme wherein the novel idea of terminal region was introduced \cite{Michalska1993}. NMPC controller is expected to drive the plant trajectory into a region, termed as terminal region or terminal set, around the origin in the finite number of steps using the feasible inputs. Subsequently, local linear controller will take the plant trajectory to the origin. This idea was extended by Chen and Allg\"ower with the concept of Quasi Infinite Horizon - Nonlinear Model Predictive Control (QIH-NMPC) formulation \cite{Chen1998}. 

One seminar contribution of \cite{Chen1998} was that it gave explicit approaches for the characterization of the terminal ingredients for the continuous time QIH-NMPC formulations. Several researchers have developed approaches for the characterization of the terminal constraints for the continuous time QIH-NMPC formulations \cite{Chen1998, Chen2003, Lucia2015, Rajhans2016} and for the discrete time QIH-NMPC formulations \cite{Limon2002, Johansen2004, Rajhans2017}. It may be noted that discrete time formulations require separate considerations due to the concept of sampling time vastly affecting the terminal region shape and size \cite{Astrom1997, Grune2011}. One of the major limitation of these approaches available in the literature is that all of them provide a single scalar tuning parameter giving only one degree of freedom for enlarging the terminal region. Recently, \cite{Yu2017} presented one approach for terminal region characterization, which is discrete time equivalent of the continuous time approach by Chen and Allg\"ower\cite{Chen1998}. However, in addition to the scalar tuning parameter limitation, Yu et al.'s approach also was suitable when linearized discrete time system at the operating point is controllable. Due to such limitation, applicability of the approach was reduced and also resulted in conservative terminal regions. 

Current work aims at extending the concepts presented by Rajhans et al. in \cite{Rajhans2017} and providing larger degrees of freedom for the terminal region characterization of the discrete time QIH-NMPC formulations.  For the proposed approaches in the current work, two tuning matrices are provided which further increase the degrees of freedom available to the designer. Proposed approaches provide three degrees of freedom namely a) linear stabilizing controller, b) additive state weighting matrix, and c) additive input weighting matrix. 

%Additionally, nominal asymptotic stability result of the continuous time QIH-NMPC formulation is presented after incorporating the updated terminal ingredients obtained after solving the modified Lyapunov equation. 

Efficacy of the proposed arbitrary controller based approach and LQR based approach is demonstrated using the benchmark two state system. This system is used by various researchers for demonstrating their stability results \cite{Mayne1990, Chen1998, Limon2005, Yu2017}.

Second section presents the discrete time QIH-NMPC formulation in detail. In addition, approach by Yu et al. \cite{Yu2017} is stated mathematically along with its limitation. Third section presents two alternate approaches for the computation of the terminal penalty and for the characterization of the terminal region.  Forth section presents numerical characterization of the terminal region using the approaches presented in the third section. Fifth section presents the simulation case study results in detail. Sixth section presents the conclusions derived.

\subsection{Notation} 
Two norm of a vector $\mathbf{y}$ is defined as $|\mathbf{y}|:= \sqrt{\mathbf{y}^T \mathbf{y}}$. Weighted norm having weighting matrix ${\mathbf{A}}$ is defined as $\mathbf{y_A} := \sqrt{\mathbf{y^T A y}}$. $norm(\mathbf{A})$ represents a two norm of matrix $\mathbf{A}$. $\lambda_{min} (\mathbf{A})$ indicates the smallest eigenvalue of $\mathbf{A}$. $\rho_{max} (\mathbf{A})$ indicates the maximum magnitude eigenvalue of $\mathbf{A}$. $\mathbb{R}$ is a set of real numbers. $\mathbb{N}$ is a set of natural numbers.  
% $\pmb{A}$ ${A}$ $\mathbold{A}$ $\bf{A}$

Consider an autonomous discrete time system 
\begin{eqnarray}
\mathbf{x}(k+1) = \tilde{\mathbf{F}} ({\mathbf{x}}(k)) \label{dsystem0}
\end{eqnarray}
with initial condition ${\bf{x}}(0) = {\mathbf{x}}_0$. Consider a set $\mathcal{X} \subset \mathbb{R}^{n_x}$ such that $\mathbf{0} \in \mathcal{X}$. 

%Definition of asymptotically stable equilibrium point is given below: 
%\begin{definition}
%The equilibrium point $\mathbf{0}$ of an autonomous discrete time system (\ref{dsystem0}) is asymptotically stable if $lim_{k \to \infty} {\mathbf{x}}(k) = {\mathbf{0}}$ for all ${\mathbf{x}}_0 \in \mathcal{X}$. 
%\end{definition}

\section{Discrete Time QIH-NMPC Formulation} 
%Consider a discrete time nonlinear system model given as  
%\begin{eqnarray}
%\mathbf{X}(k+1) = {\mathbf{F}_d} ({\mathbf{X}}(k),{\mathbf{U}}(k)) \label{dsystem1} 
%\end{eqnarray}
%where $\mathbf{X}(k) \in \mathbb{R}^{n_x}$ denotes the state
%vector in absolute terms and $\mathbf{U}(k) \in \mathbb{R}^{n_u}$ denotes the input vector in absolute terms. Let $(\mathbf{X}_s, \mathbf{U}_s)$ be the constant (time invariant) steady state of the discrete time system (\ref{dsystem1}) i.e. $\mathbf{X}_s = {\mathbf{F}_d} (\mathbf{X}_s, \mathbf{U}_s)$. Defining shift of origin as follows: 
%\begin{eqnarray}
%\mathbf{x}(k) = \mathbf{X}(k) - \mathbf{X}_s \\ 
%\mathbf{u}(k) = \mathbf{U}(k) - \mathbf{U}_s
%\end{eqnarray}
%After shift of origin, consider the discrete time nonlinear system given as  
%\begin{eqnarray}
%\mathbf{X}(k+1)-\mathbf{X}_s = {\mathbf{F}_d} ({\mathbf{X}(k)-\mathbf{X}_s},{\mathbf{U}(k)-\mathbf{U}_s}) \label{dsystem2} 
%\end{eqnarray}
%Re-wiring with a simpler notation results in  

Consider a discrete time nonlinear system model given as  
\begin{align}
\mathbf{x}(k+1) &= {\mathbf{F}} ({\mathbf{x}}(k),{\mathbf{u}}(k)) \label{dsystem} \\
{\bf{x}}(0) &= {\bf{x}}_0
\end{align}
where $\mathbf{x}(k) \in \mathcal{X} \subset \mathbb{R}^{n_x}$ denotes the state
vector and $\mathbf{u}(k) \in \mathcal{U} \subset \mathbb{R}^{n_u}$ denotes the input vector. Note that origin $(\mathbf{x}(k) = \mathbf{0}, \mathbf{u}(k) = \mathbf{0})$ is the equilibrium point of the system (\ref{dsystem}) due to shift of the origin.

Following assumptions are made: 
\begin{description}
\item[D1] System dynamics function $\mathbf{F}: \mathbb{R}^{n_x}\times \mathbb{R}^{n_u} \to \mathbb{%
R}^{n_x}$ is twice continuously differentiable. 
\item[D2] The origin $\mathbf{0} \in \mathbb{R}^{n_x}$ is an equilibrium point of the system (\ref{dsystem}) i.e. $\mathbf{F}\left( \mathbf{0}, \mathbf{0} \right) =\mathbf{0}$.
\item[D3] The inputs $\mathbf{u}(k)$ are constrained inside 
a closed and convex set $\mathcal{U} \subset \mathbb{R}^{n_u}$.
\item[D4] The system (\ref{dsystem}) has a unique solution for any initial
condition $\mathbf{x}_{0} \in \mathcal{X}$ and any input $\mathbf{u}(k) \in \mathcal{U}$. 
\item[D5] The state $\mathbf{x}(k)$ is perfectly known at any sampling instant $k$ i.e. all the states are measured. 
\item[D6] External disturbances do not affect the system dynamics. 
\end{description}

\subsection{QIH-NMPC Formulation}
For the discrete time system given by (\ref{dsystem}), QIH-NMPC formulation is stated as follows: 

\begin{equation}
\begin{array}{c}
\min \\ 
\overline{\mathbf{u}}_{(k, k+N-1)}%
\end{array}%
J\left( \mathbf{x}(k),\overline{\mathbf{u}}_{[k, k+N-1]}\right)  \label{DOptimal}
\end{equation}%
with 
\begin{eqnarray}
J\left( \mathbf{x}(k),\overline{\mathbf{u}}_{[k, k+N-1]}\right)
&=&\sum_{i=k}^{k+N-1}\left\{ 
\begin{array}{c}
\mathbf{z}(i)^T \mathbf{W}_{x} \mathbf{z}(i) \\+ \overline{\mathbf{u}}(i)^T \mathbf{W}_{u} \overline{\mathbf{u}}(i) 
\end{array}%
\right\} \notag \\
&&+\mathbf{z}(k+N)^T \mathbf{P} \mathbf{z}(k+N) \label{StageCost} 
\end{eqnarray}%
\begin{equation}
\overline{\mathbf{u}}_{[k, k+N-1]}=\left\{ \bar{\mathbf{u}}(i)\in \mathcal{U}%
: i \in \mathbb{N}, k \leq i \leq k+N-1 \right\} \label{InputSet}
\end{equation}%
subject to 
\begin{eqnarray} 
\mathbf{z}(i+1)&=&\mathbf{F}\left( \mathbf{z}(i), %
\overline{\mathbf{u}}(i)\right) \text{ for } k \leq i \leq k+N-1 \label{PredictedState} \\ 
\mathbf{z}(k)&=&\mathbf{x}(k) \label{InitialCondition} \\ 
\mathbf{z}(i) &\in& \mathcal{X} \text{ for } k \leq i \leq k+N \label{zconstraint} \\
\bar{\mathbf{u}}(i) &\in& \mathcal{U} \text{ for } k \leq i \leq k+N-1 \label{uconstraint} \\ 
\mathbf{z}\left(k+N \right) &\in& \Omega  \label{TerminalRegion}
\end{eqnarray}%
where $\mathbf{W}_{x}$ and $\mathbf{W}_{u}$ are state and input weighting matrices of dimension $\left( n_x \times n_x \right) $, $\left( n_u \times n_u \right) $ respectively. $\mathbf{P}$ is the terminal penalty matrix of dimension $\left( n_x \times n_x \right) $. $\mathbf{W}_{x}, \mathbf{W}_{u}, \mathbf{P}$ are symmetric positive definite matrices. $N$ is a finite prediction horizon length. 
%Here, $\left\Vert \mathbf{v}\right\Vert _{\mathbf{A}}^{2}=\mathbf{v}^{T}\mathbf{Av}$. 
$\mathbf{z}(i)$ denotes the predicted state in the QIH-NMPC formulation and $%
\overline{\mathbf{u}}(i)$ denotes the future input moves. The
set $\Omega$ is termed as the \emph{terminal region} in the neighborhood of
the origin. 
% and is chosen such that it is invariant for the nonlinear control system controlled by a \emph{fictitious} local linear state feedback controller with gain matrix say, $\mathbf{K}$. 
The set $\mathcal{X}_N \subset \mathcal{X} \subset \mathbb{R}^{n_x}$ is termed as the \emph{region of attraction} is the set of all feasible initial conditions i.e. it is a set of all initial conditions $\mathbf{x}_0$ such that terminal inequality constraint (\ref{TerminalRegion}) is satisfied with inputs constrained given by equation (\ref{InputSet}) satisfied.

\subsection{Design and Implementation of QIH-NMPC Formulation}
The terminal region $\Omega$ is chosen as an invariant set for the nonlinear system (\ref{dsystem}) controlled by local linear controller with gain matrix $\mathbf{K}$. The terminal penalty term is chosen such that for all trajectories starting from any point inside the terminal region $\Omega$, with approximation that a single cost term having larger value that the sum of all the predicted stage cost terms from end of horizon to infinity and is given as follows:  
\begin{equation}
\mathbf{z}(k+N)^T \mathbf{P} \mathbf{z}(k+N) \geq \sum_{i = k+N}^{\infty }\left\{ 
\begin{array}{c}
\mathbf{z}(i)^T \mathbf{W}_{x} \mathbf{z}(i)
\\ 
+ \overline{\mathbf{u}}(i) \mathbf{W}_{u} \overline{\mathbf{u}}(i)
\end{array} \right\} \label{TRCondition}
\end{equation}
with $\overline{\mathbf{u}}(i) = -\mathbf{K} \mathbf{z}(i) \in \mathcal{U}$ for all $k \geq k+N$ and for all $\mathbf{z}(k+N) \in \Omega$. 

It is assumed that the solution to the optimal problem (\ref{DOptimal}) with stage cost defined by (\ref{StageCost}) with input set given by (\ref{InputSet}) subject to the predicted state dynamics (\ref{PredictedState}) with initial condition (\ref{InitialCondition}), predicted states constraint (\ref{zconstraint}), future input moves constraint (\ref{uconstraint}) and terminal state constraint (\ref{TerminalRegion}) i.e. $\overline{\mathbf{u}}_{[k, k+N-1]}^*$ exists and can be computed numerically. Controller is implemented as a moving horizon framework. Accordingly, only the first control move $\mathbf{u}(k) = \overline{\mathbf{u}}^*(k)$ is implemented in the plant. 
%\begin{eqnarray}
%\mathbf{u}(k) = \overline{\mathbf{u}}^*(k) \label{InputMove}
%\end{eqnarray}
Entire process is repeated at the next sampling instant $k+1$. The term \emph{Quasi Infinite} is because of the fact that the NMPC formulation depicts the stability properties of the infinite horizon formulation, however, the actual implementation contains finite horizon. Such implementation is achieved with the help of the equation (\ref{TRCondition}). However, only ensuring the terminal penalty term satisfying the condition (\ref{TRCondition}) is not sufficient to guarantee the nominal asymptotic stability of the NMPC controller, hence terminal constraint as given by (\ref{TerminalRegion}) becomes inevitable. It may be noted that local linear controller with gain matrix $\mathbf{K}$ is not used for implementation of the NMPC controller and is only a mathematical construct to characterize the terminal region $\Omega$. 

\subsection{Approach from the Literature} 
Before proceeding to application of the proposed arbitrary controller based approach and LQR based approach, a look at Yu et al's approach from \cite{Yu2017} is necessary. Consider, Jacobian linearization of the nonlinear system (\ref{dsystem}) in the neighborhood the origin as,  
\begin{equation}
\mathbf{x}(k+1)=\mathbf{\Phi x}(k)+\mathbf{\Gamma u}(k)  \label{DLinSys}
\end{equation}%
where 
\begin{equation*}
\mathbf{\Phi} = \left[ \frac{\partial \mathbf{F}}{\partial \mathbf{x}}\right]
_{\left( \mathbf{0},\mathbf{0}\right) }\text{ and \ }%
\mathbf{\Gamma} = \left[ \frac{\partial \mathbf{F}}{\partial \mathbf{u}}\right]
_{\left( \mathbf{0},\mathbf{0}\right) }
\end{equation*}

One additional assumption is required at this stage. 
\begin{description}
\item[D7] The linearized system (\ref{DLinSys}) is stabilizable. 
\end{description}
Yu et al. characterize the terminal region as, 
\begin{equation}
\Omega \equiv \left\{ \mathbf{x}(k) \in \mathbb{R}^{n_x} | \mathbf{x}(k)^{T}%
\mathbf{P} \mathbf{x}(k) \leq \alpha, -\mathbf{Lx}(k) \in \mathcal{U}%
\right\}
\end{equation}%
where linear gain $\mathbf{L}$ and the terminal penalty matrix $\mathbf{P}$ are the steady state solutions of the modified Lyapunov equation given as follows:    
\begin{equation}
\kappa^2 \mathbf{\Phi}_{L}^T \mathbf{P} \mathbf{\Phi}_L - \mathbf{P} =-\mathbf{Q%
}^*  \label{YuLyapunov}
\end{equation}%
\begin{equation}
\mathbf{Q}^* = \mathbf{W}_{x} + \mathbf{L}^T \mathbf{W}_{u} \mathbf{L}  \label{Qstar}
\end{equation}
where $\mathbf{\Phi}_{L} = \mathbf{\Phi - \Gamma L}$ and parameter $\kappa$ is chosen such that $1 < \kappa < 1/ \left[ \rho_{max} \left( \mathbf{\Phi}_L \right) \right]$ with $\rho_{max}(\mathbf{\Phi}_L)$ being the maximum amplitude eigenvalue of $\mathbf{\Phi}_L$. Since the gain $\mathbf{L}$ is stabilizing, all the eigenvalues of $\mathbf{\Phi}_L$ lie inside the unit circle, hence $\rho_{max}(\mathbf{\Phi}_L) < 1$. 
It can be noted that once stage cost weighting matrices $\mathbf{W}_{x}, \mathbf{W}_{u}$ are chosen, there is barely any degree of freedom left to the designer for shaping of the terminal region. This results in  very conservative terminal regions. In addition, for the systems where there is uncontrollable eigenvalue of $\mathbf{\Phi}$ is close to unit circle, value of $\kappa$ becomes nearly $1$, resulting in a negligible margin for shaping of the terminal region. 
These limitations are overcome by using the arbitrary controller based approach wherein additive tuning matrices are introduced for providing large degrees of freedom for enlarging of the terminal region and is presented in the subsequent section.

\section{Alternate Approaches for Terminal Region Characterization} 
Two approaches are proposed in this section. First is an arbitrary controller based approach as proposed by Rajhans et al. in \cite{Rajhans2017}, and second is an LQR based approach.  
In the arbitrary controller based approach, an arbitrary stabilizing linear is designed using any of the methods available in the literature such as pole placement \cite{Kailath1980}, Linear Quadratic Gaussian (LQG) control \cite{Kirk1970} and so on. In the LQR based approach, an LQG controller is designed using the weighting matrices of the QIH-NMPC formulation. 

\begin{lemma} \label{lemma1} Suppose that assumptions D1 to D7 are satisfied and a stabilizing linear feedback control law is designed i.e. $\mathbf{\Phi}_{L}=(%
\mathbf{\Phi - \Gamma L})$ is stable indicating all the eigenvalues are inside the unit circle. Let $\Delta \mathbf{Q}$ is any positive definite matrix. Let matrix $\mathbf{P}$ denote the solution of the following modified Lyapunov equation: 
\begin{equation}
\mathbf{\Phi}_{L}^T \mathbf{P} \mathbf{\Phi}_{L} - \mathbf{P} = -(%
\mathbf{Q}^* + \Delta \mathbf{Q)}  \label{ACLyapD}
\end{equation}%
where $\mathbf{Q}^*$ is defined by equation (\ref{Qstar}). Then there exists a
constant $\alpha > 0$ which defines an ellipsoid of the form 
\begin{equation}
\Omega \equiv \left\{ \mathbf{x}(k) \in \mathbb{R}^{n_x} | \mathbf{x}(k)^{T} \mathbf{%
P} \mathbf{x}(k) \leq \alpha , -\mathbf{Lx}(k) \in \mathcal{U} \right\}
\label{ACTR}
\end{equation}%
such that $\Omega$ is an invariant set for the nonlinear system given by (\ref{dsystem}) with linear controller $\mathbf{u}(k) = - \mathbf{Lx}(k)$. Additionally, for any $\mathbf{z}(k+N) \in \Omega$ the inequality given by (\ref{TRCondition1}) holds true. 
\begin{equation}
\mathbf{z}(k+N)^T \mathbf{P} \mathbf{z}(k+N) \geq \sum_{i = k+N}^{\infty }\left\{ 
\begin{array}{c}
\mathbf{z}(i)^T \mathbf{W}_{x} \mathbf{z}(i)
\\ 
+ \overline{\mathbf{u}}(i) \mathbf{W}_{u} \overline{\mathbf{u}}(i)
\end{array} \right\} \label{TRCondition1}
\end{equation}
\end{lemma}

\begin{proof} 
Since $\mathbf{\Phi}_{L}=(%
\mathbf{\Phi - \Gamma L})$ is stable, hence, the eigenvalues of $\mathbf{\Phi}_{L}$ are inside the unit circle. Using the solvability condition of the modified Lyapunov equation, a unique $\mathbf{P} > 0$ can be computed which solves the equation (\ref{ACLyapD}). According to Assumption D2, the origin $\mathbf{0} \in \mathbb{R}^{n_u}$ is in the interior of the input constraints set $\mathcal{U}$.
Accordingly, we can compute a constant $\gamma$ which defined a set $\Omega_{\gamma}$ 
such that 
\begin{equation}
\Omega_\gamma \equiv \left\{ \mathbf{x}(k) \in \mathbb{R}^{n_x} | \mathbf{x}(k)^{T}%
\mathbf{P} \mathbf{x}(k) \leq \gamma, - \mathbf{Lx}(k) \in \mathcal{U}%
\right\} \label{Omegagamma}
\end{equation}%
Now, let $0 < \alpha \leq \gamma$ specify a region of the form given by
equation (\ref{ACTR1}). 
\begin{equation}
\Omega \equiv \left\{ \mathbf{x}(k) \in \mathbb{R}^{n_x} | \mathbf{x}(k)^{T} \mathbf{%
P} \mathbf{x} \leq \alpha \right\}
\label{ACTR1}
\end{equation}%
As the input constraints are satisfied in $\Omega
_\gamma$ and $\Omega \subseteq \Omega_\gamma$ (by virtue of $0 < \alpha \leq \gamma$), the system dynamics can be equivalently viewed as an input unconstrained system in the set $\Omega$. 
Consider a vector $\mathbf{\psi}_{L}(\mathbf{x})$ representing the nonlinearity in the system dynamics defined as 
\begin{equation}
\mathbf{\Psi}_{L}(\mathbf{x}(k))=\mathbf{F}(\mathbf{x}(k), -\mathbf{L x}(k)) - \mathbf{\Phi}_{L} \mathbf{x}(k)  \label{PsiL} 
\end{equation}%
Note for a linear system $\mathbf{\Psi}_{L}(\mathbf{x}(k)) = \mathbf{0}$. 
Consider a Lyapunov candidate defined as 
\begin{equation}
V(\mathbf{x}(k)) = \mathbf{x}(k)^{T} \mathbf{P} \mathbf{x}(k) \label{Vxk}
\end{equation}%
The difference of Lyapunov candidate $V(\mathbf{x}(k))$ can be expressed as follows: 
\begin{align}
\Delta V(\mathbf{x}(k)) &= V(\mathbf{x}(k+1)) - V(\mathbf{x}(k)) \notag \\ &= \mathbf{x}(k+1)^{T} \mathbf{P} \mathbf{x}(k+1) - \mathbf{x}(k)^{T} \mathbf{P} \mathbf{x}(k) 
\label{Vdiff1}
\end{align}%
Substituting from (\ref{PsiL}) into (\ref{Vdiff1}), 
\begin{align}
\Delta V(\mathbf{x}(k)) &= -\mathbf{x}(k)^{T} \left( \mathbf{\Phi}_{L}^{T} \mathbf{P} \mathbf{\Phi}_{L} - \mathbf{P} \right) \mathbf{x}(k) \notag \\ &+ 2 \mathbf{\Psi}_L(\mathbf{x}(k))^T \mathbf{P} \mathbf{\Phi}_{L} \mathbf{x}(k) \notag \\ &+ \mathbf{\Psi}_L(\mathbf{x}(k))^T \mathbf{P} \mathbf{\Psi}_L(\mathbf{x}(k)) 
\label{Vdiff2}
\end{align}%
Using equation (\ref{ACLyapD}) into (\ref{Vdiff2}), 
\begin{align}
\Delta V(\mathbf{x}(k)) &= -\mathbf{x}(k)^{T} \left( \mathbf{Q}^* + \Delta \mathbf{Q} \right) \mathbf{x}(k) \notag \\ &+ 2 \mathbf{\Psi}_L(\mathbf{x}(k))^T \mathbf{P} \mathbf{\Phi}_{L} \mathbf{x}(k) \notag \\ &+ \mathbf{\Psi}_L(\mathbf{x}(k))^T \mathbf{P} \mathbf{\Psi}_L(\mathbf{x}(k)) 
\label{Vdiff3}
\end{align}%
Define $\chi(\mathbf{x}(k))$ as 
\begin{align}
\chi(\mathbf{x}(k)) := & \mathbf{x}(k)^{T} \left( \Delta \mathbf{Q} \right) \mathbf{x}(k) - 2 \mathbf{\Psi}_L(\mathbf{x}(k))^T \mathbf{P} \mathbf{\Phi}_{L} \mathbf{x}(k) \notag \\ & - \mathbf{\Psi}_L(\mathbf{x}(k))^T \mathbf{P} \mathbf{\Psi}_L(\mathbf{x}(k)) 
\label{chidef1}
\end{align}%
Using equation (\ref{chidef1}) into (\ref{Vdiff3}), 
\begin{align}
\Delta V(\mathbf{x}(k)) &= -\mathbf{x}(k)^{T} \left( \mathbf{Q}^* \right) \mathbf{x}(k) - \chi(\mathbf{x}(k)) \label{Vdiff4}
\end{align}%

If $\Omega$ is chosen such that 
\begin{align}
\chi(\mathbf{x}(k)) \geq 0 \label{Vdiff13}
\end{align}%
then 
\begin{align}
\Delta V(\mathbf{x}(k)) \leq -\mathbf{x}(k)^{T} \left( \mathbf{Q}^* \right) \mathbf{x}(k)  \label{Vdiff14}
\end{align}%
%Equation (\ref{Vdiff14}) for inequality based method is identical to equation (\ref{Vdiff7}) for norm based method. 

Summing the inequality (\ref{Vdiff14}) over the interval, $[k+N, \infty
),$ it follows that 
\begin{equation}
V(\mathbf{x}(k+N)) \geq \sum_{i=k+N}^{\infty} \mathbf{x}(i)^{T}%
\mathbf{Q}^* \mathbf{x}(i) \label{Vdiff15}
\end{equation}%
Using (\ref{Qstar}) and (\ref{Vxk}), inequality (\ref{Vdiff15}) becomes identical to (\ref{TRCondition1}) and holds true for any $\mathbf{x}(k+N)) \in \Omega$.
\end{proof}

\begin{lemma} \label{lemma2} Suppose that assumptions D1 to D7 are satisfied. Let $\widetilde{\mathbf{W}}_x > \mathbf{W}_x$ and $\widetilde{\mathbf{W}}_u > \mathbf{W}_u$ be positive definite matrices. Let matrix $\mathbf{L}$ and $\mathbf{P}$ denote the steady state solution of the following modified Lyapunov equations: 
\begin{align}
\mathbf{\Phi}_{L}^T \mathbf{P} \mathbf{\Phi}_{L} - \mathbf{P} = -\left(%
\widetilde{\mathbf{W}}_x + \mathbf{L}^T \widetilde{\mathbf{W}}_u \mathbf{L} \right)  \label{LQRLyapD1} \\ 
\mathbf{L} = \left( \widetilde{\mathbf{W}}_u + \mathbf{\Gamma}^T \mathbf{P} \mathbf{\Gamma} \right)^{-1} \mathbf{\Gamma}^T \mathbf{P} \mathbf{\Phi} 
\label{LQRLyapD2}  
\end{align}%
Then there exists a constant $\alpha > 0$ which defines an ellipsoid of the form (\ref{ACTR})
such that $\Omega$ is an invariant set for the nonlinear system given by (\ref{dsystem}) with linear controller $\mathbf{u}(k) = - \mathbf{Lx}(k)$. Additionally, for any $\mathbf{z}(k+N) \in \Omega$ the inequality given by (\ref{TRCondition1}) holds true. 
\end{lemma}
\begin{proof} 
Define the following matrices, 
\begin{equation}
\Delta \mathbf{Q} = \Delta \mathbf{W}_{x} + \mathbf{L}^T \Delta \mathbf{W}_{u} \mathbf{L}  \label{DeltaQ}
\end{equation}
with 
\begin{align}
\Delta \mathbf{W}_x := \widetilde{\mathbf{W}}_x - \mathbf{W}_x \label{DeltaWx} \\ 
\Delta \mathbf{W}_u := \widetilde{\mathbf{W}}_u - \mathbf{W}_u \label{DeltaWu}  
\end{align}
Consider a Lyapunov candidate defined as 
\begin{equation}
V(\mathbf{x}(k)) = \mathbf{x}(k)^{T} \mathbf{P} \mathbf{x}(k) \label{Vxk1}
\end{equation}%
The difference of Lyapunov candidate $V(\mathbf{x}(k))$ can be expressed as follows: 
\begin{align}
\Delta V(\mathbf{x}(k)) &= V(\mathbf{x}(k+1)) - V(\mathbf{x}(k)) \notag \\ &= \mathbf{x}(k+1)^{T} \mathbf{P} \mathbf{x}(k+1) - \mathbf{x}(k)^{T} \mathbf{P} \mathbf{x}(k) 
\label{Vdiff21}
\end{align}%
Substituting from (\ref{LQRLyapD1}), (\ref{DeltaQ}), (\ref{DeltaWx}), and (\ref{DeltaWu}) into (\ref{Vdiff21}) and upon simplification, 
\begin{align}
\Delta V(\mathbf{x}(k)) \leq -\mathbf{x}(k)^{T} \left( \mathbf{Q}^* \right) \mathbf{x}(k)  \label{Vdiff22}
\end{align}%
Rest of the proof is similar to that of the proof of the lemma \ref{lemma1}. 
\end{proof}

\begin{lemma} \label{lemma3} Let the assumptions D1-D7 hold true. For the nominal discrete time system, feasibility of QIH-NMPC formulation problem (\ref{DOptimal}) at sampling instant $k=0$ implies its feasibility for all $k > 0$.  
\end{lemma}
\begin{proof}
Proof is identical to the proof of the lemma 2 from \cite{Rajhans2017}. 
\end{proof}

\begin{theorem} \label{theorem1} Let a) Assumptions D1-D7 hold true and b) the discrete time QIH-NMPC problem be feasible at $k = 0$. The nominal nonlinear system (\ref{dsystem}) controlled with QIH-NMPC controller is asymptotically stable at the origin. 
\end{theorem}

\begin{proof}
From equation (\ref{Vxk}) from the lemma (\ref{lemma1}), 
consider the Lyapunov candidate function 
\begin{equation}
V(\mathbf{x}(k)) = \mathbf{x}(k)^{T} \mathbf{P} \mathbf{x}(k) \label{Vx1}
\end{equation}%
Consider the following three properties \cite{Khalil2002}: 
\begin{itemize}
\item $V(\mathbf{0}) = (\mathbf{0}^{T}) \mathbf{P} (\mathbf{0}) = 0$. 
\item Since $\mathbf{P}$ is a positive definite matrix, $V(\mathbf{x}(k)) = \mathbf{x}(k)^{T} \mathbf{P} \mathbf{x}(k) > 0$ for all $\mathbf{x}(k) \neq \mathbf{0}$. 
\item Using (\ref{Vdiff14}) and $\mathbf{Q}^* > 0$ implies 
\begin{align}
\Delta V(\mathbf{x}(k)) \le -\mathbf{x}(k)^{T} \mathbf{Q}^* \mathbf{x}(k) < 0 \label{Vx2}
\end{align}%
\end{itemize}

Thus, the candidate function $V(\mathbf{x}(k))$ is a Lyapunov function for the
nonlinear system for $\mathbf{x} \in \Omega$ under QIH-NMPC controller. 
Hence, the closed loop system is asymptotically stable at the origin. 
%Detailed steps of the proof are similar to the proof of the theorem 1 from \cite{Chen1998}. 
\end{proof}

\section{CHARACTERIZATION OF THE TERMINAL REGION}
Lemma \ref{lemma1} or lemma \ref{lemma2} gave the conditions for explicit characterization of the terminal region. It is possible to compute the terminal region numerically and subsequently implement the QIH-NMPC controller. 
% Let $\mathbf{u} = u_1, u_2, ..., u_{n_u}$ where $u_{i} \in \mathbb{R}$ for $i = 1, 2, ..., n_u$.
%\subsection{Steps for the Characterization of the Terminal Region} 

Steps for characterization of the terminal region using arbitrary controller based approach and LQR based approach are given below: 

\begin{description}
\item[S1] Computation of Upper Bound Set: \\ 
Compute the largest value of $\gamma$ such that inputs constraints are satisfied in the set $\Omega_\gamma$ given by (\ref{Omegagamma1}). 
\begin{equation}
\Omega_\gamma \equiv \left\{ \mathbf{x}(k) \in \mathbb{R}^{n_x} | \mathbf{x}(k)^{T}%
\mathbf{P} \mathbf{x}(k) \leq \gamma, - \mathbf{K x}(k) \in \mathcal{U}%
\right\} \label{Omegagamma1}
\end{equation}%
This can be formulated as a Quadratic Programming (QP) problem if the constraints are defined by upper bound and lower bound on each of the input variables. Typically the set $\Omega_\gamma$ would be tangential to at least one of the input constraints. 
\item[S2] Computation of the Terminal Region using inequality based method: \\ 
Compute the largest $\alpha \in (0, \gamma]$ such that 
\begin{align}
\left[ 
\begin{array}{c} 
\min \\ 
\mathbf{x}(k) \in \Omega%
\end{array}%
\chi(\mathbf{x}(k)) \right] = 0 \label{TRCompute3}
\end{align}
The condition given by (\ref{TRCompute3}) ensures that $\mathbf{\Psi}(\mathbf{x}) > 0$ for all $\mathbf{x} \in \Omega$, which is the necessary condition to further establish the nominal asymptotic stability. 
\end{description} 
The step S2 is implemented as follows: \\
Initially $\alpha = \gamma$ and condition (\ref{Vdiff13}) i.e. $(\chi(\mathbf{x}(k)) \geq 0)$ is checked. If (\ref{Vdiff13}) is true, then $\alpha = \gamma$. If (\ref{Vdiff13}) is false i.e. $\chi(\mathbf{x})(k) < 0)$ for at least one $\mathbf{x}(k) \in \Omega$, then the value of $\alpha$ is further reduced by a multiplicative factor $\beta < 1$ and $\beta \approx 1$. The process continues until condition (\ref{Vdiff13}) is satisfied. 
 
%\subsection{Quantification of the Terminal Region}
Terminal region shape changes according to the computed $\mathbf{P}$ matrix and its size changes according to the value of $\alpha$. Comparison of the terminal regions obtained using various approaches is carried out using measurement of the area for the system with a state dimension of 2. Area of the terminal region $\Omega$ defined by (\ref{ACTR}) is given by 
\begin{align}
A_2 = \frac{\pi \alpha}{\sqrt{det(\mathbf{P})}} \label{Area}
\end{align}

\section{Simulation Case Study} 
Effectiveness of the proposed approaches for the terminal region characterization and its applicability to QIH-NMPC discrete time simulations is demonstrated using the benchmark two state system which is used by several researchers \cite{Mayne1990, Chen1998, Yu2017}. 

\subsection{Choice of Tuning Matrices}
According to the design of the arbitrary controller based approach, the gain matrix $\mathbf{K}$ can be any arbitrary stabilizing linear controller. However, in order to simplify the computations, simulation results are presented with the following choice. Identical linear controller gain is also used in the case of LQR based approach. Controller gain $\mathbf{K}$ is the steady state solution of the simultaneous equations (\ref{LQRP}) and (\ref{LQRK}).  
\begin{align}
\mathbf{\Phi}_{L}^T \mathbf{P} \mathbf{\Phi}_{L} - \mathbf{P} = -\left(%
\mathbf{W}_x + \mathbf{L}^T \mathbf{W}_u \mathbf{L} \right)  \label{LQRP} \\ 
\mathbf{L} = \left( \mathbf{W}_u + \mathbf{\Gamma}^T \mathbf{P} \mathbf{\Gamma} \right)^{-1} \mathbf{\Gamma}^T \mathbf{P} \mathbf{\Phi} 
\label{LQRK}  
\end{align}%

The tuning matrix $\Delta \mathbf{Q}$ for the arbitrary controller based approach can be any positive definite matrix. However, in order to simplify and structure the computations of the terminal region, following parameterization is carried out:   
\begin{equation}
\Delta \mathbf{Q} = \widetilde{\mathbf{W}}_{x} + \mathbf{K}^T \widetilde{\mathbf{W}}_{u} \mathbf{K}  \label{ACDQ}
\end{equation}%
In order to further simplify the numerical computation of the terminal region using arbitrary controller based approach, additional parameterization is carried out and is given as follows: 
\begin{equation}
\widetilde{\mathbf{W}}_{x} = \rho_x \mathbf{W}_{x} \text{ and } \widetilde{\mathbf{W}}_{u} = \rho_u \mathbf{W}_{u} \label{TuningM}
\end{equation}
Note that it is sufficient to have $\widetilde{\mathbf{W}}_{x} > \mathbf{W}_{x}$ or $\widetilde{\mathbf{W}}_{u} > \mathbf{W}_{u}$ to satisfy $\Delta \mathbf{Q} > 0$, however, usually both $\widetilde{\mathbf{W}}_{x} > \mathbf{W}_{x}$ and $\widetilde{\mathbf{W}}_{u} > \mathbf{W}_{u}$ is preferred in practice.  
Using the matrices (\ref{TuningM}) into (\ref{ACDQ}), 
\begin{equation}
\Delta \mathbf{Q} = \rho_x \mathbf{W}_{x} + \rho_u \mathbf{K}^T \mathbf{W}_{u} \mathbf{K}  \label{ACDQ1}
\end{equation}
where $\rho_x$ and $\rho_u$ are the tuning scalars. 
For the arbitrary controller based approach, conditions $\rho_x > 0$, and $\rho_u > 0$ and for the LQR based approach, conditions $\rho_x > 1$ and $\rho_u > 1$ are imposed. 

Efficacy of having two tuning parameters is efficiently demonstrated using the case study in the next sub-section. In the case study, in the first iteration, parameter $\rho_x$ is increased from $0$ (for arbitrary controller based approach) or $1$ (for LQR based approach) to a large value, keeping $\rho_u$ constant at $0$ or $1$ respectively. A value of $\rho_x$ is chosen which leads to larger terminal region in the first iteration. Subsequently, in the second iteration, parameter $\rho_u$ is increased from $0$ (for arbitrary controller based approach) or $1$ (for LQR based approach) to larger value, keeping $\rho_x$ as constant to the values obtained in the first iteration.

\subsection{Two State System} 

Discrete time system dynamics are given by (\ref{CAEeq1})-(\ref{CAEeq2}). These equations are obtained by using the continuous time equations given in \cite{Mayne1990, Chen1998} with a simple first order Euler integration to obtain the discrete time equations identical to the ones given in \cite{Yu2017}. 
\begin{align}
x_1 (k+1) & = x_1(k) + T \left( x_2 (k) + u(k) \left( \mu_0 + (1-\mu_0) x_1(k) \right) \right) \label{CAEeq1} \\
x_2 (k+1) & = x_2(k) + T \left( x_1 (k) + u(k) \left( \mu_0 - 4 (1-\mu_0) x_2(k) \right) \right) \label{CAEeq2}
\end{align}%
where parameter $\mu_0 = 0.5$, $\mathbf{x}(k) = \left[x_1(k) ~~ x_2(k) \right]^T$ is the state vector, and $u(k)$ is the input. The chosen equilibrium operating point is given as $\mathbf{x}(k) = \left[0 ~~ 0 \right]^T$ and $\mathbf{u}(k) = \left[ 0 \right]$. 
%follows:%
%\begin{equation}
%\mathbf{X}_{s}=\left[ 
%\begin{array}{cc}
%0 & 0 
%\end{array}%
%\right]^T \text{ and }\mathbf{U}_{s}=\left[ 
%\begin{array}{c}
%0
%\end{array}%
%\right]
%\end{equation}%
It may be noted that sampling interval chosen for computations is $T = 0.1$s, which is identical to the one chosen in \cite{Yu2017}. The eigenvalues of the linearized matrix at the origin are $(0.9048, 1.1052)$, indicating that the chosen operating point is unstable resulting in a challenging control problem. 

The discrete time QIH-NMPC weighting matrices are chosen identical to the one given in \cite{Yu2017} and are given as follows: 
\begin{align}
\mathbf{W}_{x} = \left[ \begin{array}{ccc}
 1 & 0 \\ 
0 & 1 
\end{array} \right] \text{ and } 
\mathbf{W}_{u} = \left[ \begin{array}{cc}
0.5 
\end{array} \right] \label{CAECWM}
\end{align}%
The input constraint used in the QIH-NMPC formulation is identical to the one given in \cite{Chen1998, Yu2017} and is given as follows: 
\begin{equation}
\mathcal{U}=\left\{ \mathbf{u}(k)\in \mathbb{R}~|~-2 \leq \mathbf{u}(k)\leq 2 \right\} 
\end{equation}
%
%Terminal region obtained using Chen and Allg\"ower's inequality approach is\ $\Omega
%=\left\{ \mathbf{x}^{T}\mathbf{P} \mathbf{x} \leq 2.1255 \times 10^{-4} \right\} $ with linear gain matrix and terminal penalty matrix computed as follows: 
%\begin{equation}
%\mathbf{K}_0 = \left[ 
%\begin{array}{ccc}
%   -0.2994  &  0.7249 &  -1.1771 \\ 
%    0.0774   & 0.5753  &  0.3908
%\end{array}%
%\right] 
%\end{equation}
%\begin{equation}
%\mathbf{P}_0 = \left[ 
%\begin{array}{ccc}
%   59.8750 &  -3.2848 & -16.0747 \\
%   -3.2848  &  5.5392 &   4.1659 \\
%  -16.0747  &  4.1659 &   6.3596 \\
%\end{array}%
%\right] 
%\end{equation}

Table {\ref{CAE_TR_Disc_AC_rhox}} presents the terminal regions obtained using the arbitrary controller based approach when only one tuning parameter $\rho_x$ is varied for the two state system. 

\begin{table}[tbph]
\caption{Terminal Region for arbitrary controller based approach for Two State System for Varying $\rho_x$ only}
\label{CAE_TR_Disc_AC_rhox}\centering%
\begin{tabular}{|c|c|c|c|}
\hline 
$\rho_x$ & $\gamma$ & $\alpha$ & Area \\ \hline \hline
    0.1 &   9.8803 &   0.1537 &   0.0391 \\ \hline  
    1 &  11.3196  &  0.1700  &  0.0300 \\ \hline  
    5 &  17.7167 &   0.7247  &  0.0590 \\ \hline  
   20 &  41.7057 &   4.1387  &  0.1174 \\ \hline  
   50 &  89.6837 &  10.6782  &  0.1325 \\ \hline  
  100 & 169.6470 &  22.2603  &  0.1428 \\ \hline  
\end{tabular}%
\end{table}

From the table {\ref{CAE_TR_Disc_AC_rhox}}, it may be observed that the terminal region area has nearly saturated at $\rho_x = 100$ and is used as the value for the second iteration. Table {\ref{CAE_TR_Disc_AC_rhou}} presents the terminal regions obtained using the arbitrary controller based approach when the second tuning parameter $\rho_u$ is varied, keeping first tuning parameter $\rho_x$ as constant. It may be noted that in this particular system, value of $\alpha$ which defines the terminal region boundary are found identical to its upper bound $\gamma$ and it probably because of the mild nonlinearity of the system dynamics near the origin. 

\begin{table}[tbph]
\caption{Terminal Region for arbitrary controller based approach for Two State System for Varying $\rho_u$ only}
\label{CAE_TR_Disc_AC_rhou}\centering%
\begin{tabular}{|c|c|c|c|}
\hline 
$\rho_u$ & $\gamma$ & $\alpha$ & Volume \\ \hline \hline
   1 &   177.8  &  22.6 &   0.1416 \\ \hline  
   5 &   210.3  &  28.2 &   0.1626 \\ \hline  
   10 &   250.9  &  33.2 &   0.1750 \\ \hline  
  100 &   981.8 &   106.5 &   0.2838 \\ \hline  
  200 &   1793.9 &   131.2 &   0.2588 \\ \hline  
\end{tabular}%
\end{table} 

It can be noticed from the table \ref{CAE_TR_Disc_AC_rhou} that the terminal region area has reached at $\rho_u = 100$.

Similar to the arbitrary controller based approach, LQR based approach is implemented. Table {\ref{CAE_TR_Disc_LQR_rhox}} presents the terminal regions obtained using the LQR based approach when only one tuning parameter $\rho_x$ is varied and using the inequality method for the two state system. 

\begin{table}[tbph]
\caption{Terminal Region for LQR based approach for Two State System for Varying $\rho_x$ only}
\label{CAE_TR_Disc_LQR_rhox}\centering%
\begin{tabular}{|c|c|c|c|}
\hline 
$\rho_x$ & $\gamma$ & $\alpha$ & Area \\ \hline \hline
    5 &   7.9304 &   0.2574  &  0.0253 \\ \hline 
   20 &   6.6979 &   1.2780 &   0.0467 \\ \hline 
   50 &   6.6179 &   3.1971 &  0.0575 \\ \hline 
  100 &   7.2567 &   4.7021 &   0.0479 \\ \hline 
\end{tabular}%
\end{table}

From the table {\ref{CAE_TR_Disc_LQR_rhox}}, terminal region with largest area is obtained at $\rho_x = 50$, which is the value used for the second iteration. Table {\ref{CAE_TR_Disc_LQR_rhou}} presents the terminal regions obtained using the LQR based approach when the second tuning parameter $\rho_u$ is varied, keeping first tuning parameter $\rho_x$ as constant. 

\begin{table}[tbph]
\caption{Terminal Region for LQR based approach for Two State System for Varying $\rho_u$ only}
\label{CAE_TR_Disc_LQR_rhou}\centering%
\begin{tabular}{|c|c|c|c|}
\hline 
$\rho_u$ & $\gamma$ & $\alpha$ & Volume \\ \hline \hline
   1.1 &  7.832    &  5.159    & 0.0517 \\ \hline 
   5    &  33.490  & 17.369   & 0.1269 \\ \hline 
   10  &  72.006  &  27.836  & 0.1689 \\ \hline 
   20  &  158.61  &  42.685  & 0.2093 \\ \hline 
   50  &  450.21  &  71.597  & 0.2542 \\ \hline 
   100 & 972.03  &   95.677 & 0.2576 \\ \hline 
   200 & 2051.1  & 125.50   & 0.2501  \\ \hline 
\end{tabular}%
\end{table} 

It can be noticed from the table \ref{CAE_TR_Disc_LQR_rhou} that the terminal region area has reached maximum saturation around $\rho_u = 100$.

Table {\ref{CAE_TR_Disc_Compare}} presents a comparison of the terminal regions obtained using approach given in \cite{Yu2017}, Arbitrary Controller (AC) based approach, and LQR based approach for the two state system system. Table presents the tuning parameters, values of $\gamma$ (upper bound on $\alpha$) and $\alpha$ (which defines the terminal region boundary) and the area of each of the terminal regions. 

\begin{table}[tbph]
\caption{Comparison of Terminal Regions for Two State System using Various Approaches}
\label{CAE_TR_Disc_Compare}\centering%
\begin{tabular}{|c|c|c|c|c|c|}
\hline 
Approach & Parameter & $\gamma$ & $\alpha$  &  Area \\ \hline \hline
Yu et al. \cite{Yu2017} & $ \kappa = 1/0.91 $ &  123.66 & 0.778 & 0.0271 \\ \hline 
AC & $ \rho_x = 100 $, $ \rho_u = 0 $ & 169.65 & 22.26 & 0.1428 \\ \hline 
AC & $ \rho_x = 100 $, $ \rho_u = 100 $ & 981.8 & 106.5 & 0.2838 \\ \hline 
LQR & $ \rho_x = 50 $, $ \rho_u = 1 $ & 6.6179 & 3.1971 & 0.0575 \\ \hline 
LQR & $ \rho_x = 50 $, $ \rho_u = 100 $ & 972.03 & 95.677 & 0.2576 \\ \hline 
%$1 \times 10^{6}$  & $ 105.79 $ & 105.79 & $ 1.1029 \times 10^{-3}$ \\ \hline 
\end{tabular}%
\end{table} 

From the table {\ref{CAE_TR_Disc_Compare}}, it can be observed that the largest terminal regions obtained using the arbitrary controller based approach and LQR based approach are approximately 10.4723 and 9.5055 times larger by area respectively when compared to the largest terminal region obtained using the approach given by Yu et al. in \cite{Yu2017}. It also illustrates the significant impact of the second degree of freedom in the form of tuning parameter $\rho_u$ (or tuning parameter matrix $\widetilde{\mathbf{W}}_u$) on the size of the terminal region.

The terminal region obtained using the arbitrary controller based approach (with $\rho _{x} = 100$ and $\rho _{u} = 100$) having largest area is\ $\Omega
=\left\{ \mathbf{x}^{T}\mathbf{P} \mathbf{x} \leq 106.5 \right\} $ with an area of $0.2838$. Corresponding values of linear gain matrix $\mathbf{K}$ and terminal
penalty matrix $\mathbf{P}$ are given as follows:%
\begin{equation}
\mathbf{K} = \left[ 
\begin{array}{cc}
    2.2534  &  2.2534 
\end{array}%
\right] 
\end{equation}
\begin{equation}
\mathbf{P} = 10^3 \times \left[ 
\begin{array}{cc}
    1.5249  &  0.9677 \\ 
    0.9677  &  1.5249  
\end{array}%
\right] 
\end{equation}

The terminal region obtained using the LQR based approach (with $\rho _{x} = 50$ and $\rho _{u} = 100$) having largest area is\ $\Omega
=\left\{ \mathbf{x}^{T}\mathbf{P} \mathbf{x} \leq 95.677 \right\} $ with an area of $0.2576$. Corresponding values of linear gain matrix $\mathbf{K}$ and terminal
penalty matrix $\mathbf{P}$ are given as follows:%
\begin{equation}
\mathbf{K} = \left[ 
\begin{array}{cc}
    2.2534 &   2.2534 
\end{array}%
\right] 
\end{equation}
\begin{equation}
\mathbf{P} = 10^3 \times \left[ 
\begin{array}{cc}
    1.5098  &  0.9582 \\ 
    0.9582  &  1.5098  
\end{array}%
\right] 
\end{equation}
Note that the gain matrix is identical in both the approaches by choice. Discrete time QIH-NMPC controller simulations are carried out using MATLAB software. Sampling interval used for computer simulations is $0.1$ seconds. Initial condition for the state is identical to the one given in \cite{Yu2017} and is given as $\mathbf{x}_0 = [-3~~2]^T$.

Table \ref{NMPC_Disc_Pred_CAE} gives minimum prediction horizon length required to satisfy the terminal constraint. It can be observed that the minimum prediction horizon required in the case of arbitrary controller based approach and LQR based approach is significantly smaller as compared to the prediction horizon required in the case of terminal constraint obtained using Yu et al.'s approach. 
%In addition, the fact that larger the terminal region, smaller is the minimum prediction horizon required for satisfaction of the inequality terminal constraint for identical initial conditions. 
It is well established in the literature that the computation time required for NMPC iterations increases exponentially with the prediction horizon length \cite{Mayne2000, Chen1998, Rawlings2017}. Results obtained effectively validate the advantage of having larger terminal regions.

\begin{table}[tbph]
\caption{Minimum Prediction Horizon Length Required for Various Initial Conditions in the Two State System}
\label{NMPC_Disc_Pred_CAE}\centering%
\begin{tabular}{|c|c|}
\hline 
Approach & $N$ \\ \hline \hline
Yu et al.'s \cite{Yu2017} approach & 28 \\ \hline 
Arbitrary controller based approach & 11 \\ \hline 
LQR based approach & 12 \\ \hline 
\end{tabular}%
\end{table}

Figure \ref{CAE_Disc_Compare_TR} depicts the plot of the boundaries of the terminal regions obtained using three approaches namely a) approach given by Yu et al. in \cite{Yu2017}, b) LQR based approach, and c) Arbitrary controller based approach for the two state system for an identical sampling interval of $0.1$s. It may be noted that the terminal regions obtained using arbitrary controller based approach and LQR based approach are significantly larger than the approach by Yu et al. In addition there are several states which lie inside or near to the larger terminal regions obtained using the proposed approaches resulting in feasible initial conditions with a very small prediction horizon length, which will require a relatively larger prediction horizon lengths when the terminal region given in \cite{Yu2017} is used as terminal inequality constraint.  

\begin{figure}[!ht]
\centerline{\includegraphics[width=\columnwidth]{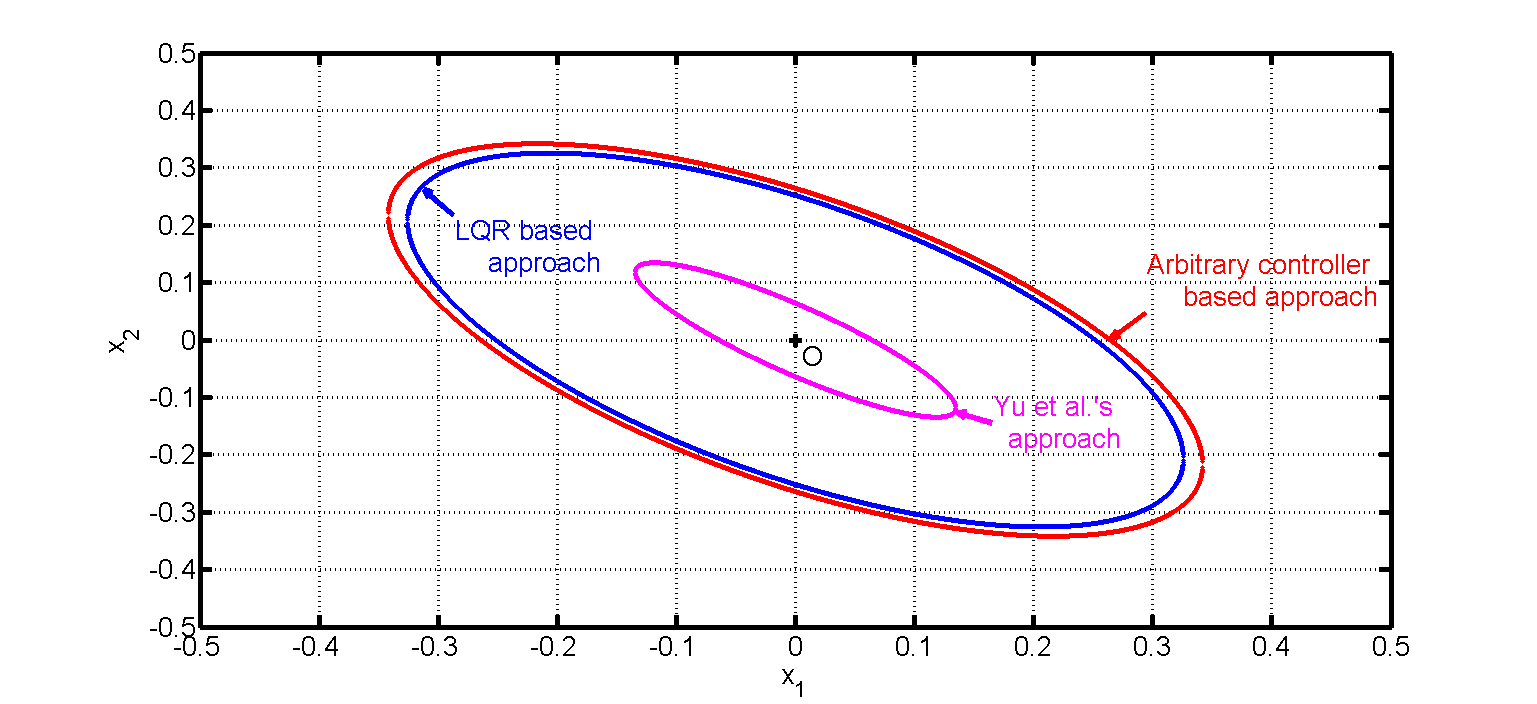}}
\caption{Two states system: Graphical Comparison of the Terminal Regions}
\label{CAE_Disc_Compare_TR}
\end{figure}

\subsection{Conclusions} 
Approaches presented in the literature for the terminal region characterization for the discrete time QIH-NMPC formulations provide limited degrees of freedom and often result in a conservative terminal region. Two approaches are presented in this work namely an arbitrary controller based approach and an LQR based approach. Proposed approaches provide two additive matrices as tuning parameters for enlargement of the terminal region. Both the approaches are scalable to system of any state and input dimension. 

Efficacy of the approaches is demonstrated using benchmark two state system. It is observed that the terminal region obtained using the arbitrary controller based approach and LQR based approach are approximately 10.4723 and 9.5055 times larger by area measure when compared to the largest terminal region obtained using the approach given by Yu et al. \cite{Yu2017}. Future research would involve choosing a completely arbitrary linear controller gain and parameterizing the approaches using different combinations of additive tuning matrices.

\bibliography{Journal3_bib1}

\end{document}